%% file: main_iv.tex
\documentclass[runningheads]{llncs}
\usepackage{graphicx}
\usepackage{colordvi}
\usepackage{color}
\usepackage{paralist}
\usepackage{bm}
\usepackage{hyperref}
\usepackage{csquotes}
\usepackage{mathpazo}
\usepackage{epsfig}
\usepackage{amsmath,amssymb}
\usepackage{amsfonts}
\usepackage{amssymb}
\usepackage{amstext}
\usepackage{amsmath}
\usepackage{xspace}
\usepackage{theorem}
\usepackage{thmtools}
\usepackage{mathtools}
\usepackage{algorithm}
\usepackage[noend]{algpseudocode}
\usepackage{sidecap}
\usepackage{caption}

\newcommand{\remove}[1]{}

\newcommand{\mcP}{\mathcal{P}}

\newcommand{\mcC}{\mathcal{C}} 
\newcommand{\mcL}{\mathcal{L}}

\newcommand{\junk}[1]{}
\newenvironment{lp}[2]{\[\begin{array}{rcll}
                        \mbox{#1} & & #2 \\ 
                        \mbox{subject to}}
                         {\end{array}\]}
\newcommand{\cnstr}[4]{\\ #1 & #2 & #3 & #4}

\begin{document}

\title{Integer Plane Multiflow Maximisation:\\Flow-Cut Gap and One-Quarter-Approximation}
\titlerunning{Integer Plane Multiflow Maximisation: Gaps and Approximation}
%
\author{Naveen Garg, Nikhil Kumar, Andr\'as Seb\H{o}}
\authorrunning{N. Garg, N. Kumar, A. Seb\H{o}}
%
\institute{Indian Institute of Technology, Delhi\\
\email{naveen@cse.iitd.ac.in}\\ 
\email{nikhil@cse.iitd.ac.in}\\
CNRS, Laboratoire G-SCOP, Univ. Grenoble Alpes\\
\email{andras.sebo@grenoble-inp.fr}}
\maketitle              
\begin{abstract}
In this paper, we bound the integrality gap  and the approximation ratio  for   maximum plane multiflow problems  and deduce  bounds on the flow-cut-gap.  Planarity means here that the union of the supply and demand graph  is planar.   
We first prove that there exists a multiflow of value at least half of the capacity of a minimum multicut. We then show how to convert any multiflow into a half-integer one of value at least half of the original multiflow. Finally, we  round any half-integer multiflow into an integer multiflow, losing again at most half of the value, in polynomial time, achieving a $1/4$-approximation algorithm for maximum integer multiflows in the plane, and an integer-flow-cut gap of $8$. 
\end{abstract}

\input{introduction_iv.tex}

\input{preliminaries_iv.tex}

\input{multicut-2-connectivity_iv.tex}
\input{frac_half_iv.tex}
\input{integralflow_iv.tex}

\input{integralitygap_iv}

\input{conclusions_iv.tex}

\bibliographystyle{plain}
\bibliography{sp_iv.bib}
\input{app_iv.tex}	
\end{document}

%% file: introduction_iv.tex
\section{Introduction}\label{sec:introduction} 

Given an undirected graph $G=(V,E)$ with edge capacities  $c:E\rightarrow\mathbb{R}_+$, and some pairs of vertices given as edges of the graph $H=(V,F)$, the {\em maximum-multiflow problem} with input $(G,H,c)$, asks for the maximum flow that can be routed in $G$, simultaneously between the endpoints of edges in $F$, respecting the capacities $c$. 

  This is one of many widely studied variants of the multiflow problem. Other popular variants include demand flows, all or nothing flows, unsplittable flows etc. In this paper, we are mainly interested in the integer version of this problem, and the half-integer or fractional versions also occur as tools. When capacities are $1$, the capacity constraint specialises to edge-disjointness, whence the {\em maximum edge disjoint paths problem (MEDP)} between a given set of pairs is a special case; MEDP is NP-Hard to compute for general graphs, even in very restricted settings like when $G$ is a tree \cite{garg1997primal}. 

The edges in $F$ are called {\em demand edges} (sometimes commodities), those in $E$ are called {\em supply edges}; accordingly, $H=(V,F)$ is  the {\em demand graph}, and $G=(V,E)$ is the {\em supply graph}. If $G+H=(V,E\cup F)$ is planar we call the problem a {\em plane} multiflow problem.  Plane multiflows have been studied for the past forty years, starting with Seymour \cite{seymour1981odd}.  

Let $\mcP_e$ $(e\in F)$ be the set of  paths in $G$ between the endpoints of $e$, and $\mcP:=\cup_{e\in F}\mcP_e$. For $P\in\mcP_e$, the edge $e$ is said to be the {\em demand-edge of $P$}, denoted by $e_P$.   A {\em multiflow}, or for simplicity a {\em flow} in this paper is a function $f:\mcP\rightarrow\mathbb{R}^+$. 
The flow $f$ is called {\em feasible}, if $\sum_{\{P\in\mcP:e\in P\}} f(P)\le c(e)$ for all $e\in E$. The {\em value} of a flow $f$ is defined as $|f|:=\sum_{P\in\mcP} f(P)$.

For a path $P\in\mcP$, we refer to $f(P)$ as the {\em flow on $P$}. If the flow on every path is integer or half-integer, we say that the flow is integer or half-integer, respectively. A {\em circuit} is a connected subgraph with all degrees equal to two.

{\em The value of multiflows (without restrictions on integrality) can be maximised in general, in (strongly) polynomial time} \cite[ 70.6, page 1225]{schrijver2003combinatorial} using a linear programming algorithm.

 A {\em multicut} for $(G,H)$ is a set of edges $M\subseteq E$ such that every $P\in\mcP$ contains at least one edge in $M$. \footnote{Given a partition of $V$ so that all demand edges join different classes, the supply edges that join different  classes form a multicut, and inclusionwise minimal multicuts are of this form.}   Multicuts provide the simplest possible and most natural constraints for the dual of the maximum multiflow problem. 
 
 It is easy to see that the value of any feasible multiflow is smaller than or equal to the capacity of any multicut. Klein, Mathieu and Zhou \cite{klein2015correlation}  prove  that computing the minimum multicut is NP-hard if $G+H$ is planar,  and they also provide a PTAS.
 
 


There is a rich literature on the maximum ratio of a minimum capacity of a multicut over the maximum multiflow. With an abuse of terminology we will call this the  {\em (possibly integer or half-integer) flow-cut gap}, sometimes also restricted to subclasses of problems. This is not to be confused with the same term used for demand problems.  The integer flow-cut gap is $1$ when $G$ is a path and $2$ when $G$ is a tree. For arbitrary $(G,H)$, the flow-cut gap is $\theta(\log |F|)$ \cite{garg1996approximate}. Building on decomposition theorems from Klein, Plotkin and Rao \cite{klein1993excluded}, Tardos and Vazirani~\cite{TV} showed a flow-cut gap of $O(r^3)$ for graphs which do not contain a $K_{r,r}$ minor; note that for $r=3$ this includes the class of planar graphs. A long line of impressive work, culminated in \cite{seguin2011maximum} proving a constant approximation ratio for maximum half-integer flows, which together with ~\cite{TV} implies a constant half-integer flow-cut gap for planar supply graphs. A simple topological obstruction proves that the integer flow-cut gap for planar supply graphs, even when all demand edges  are on one face of the graph,  also called {\em Okamura-Seymour instances},  is $\Omega(|F|)$\cite{garg1997primal}. 

These results are often bounded by the {\em integrality gap} of the multiflow problems for the problem classes, which is the infimum of MAX$/|f|$ over all instances $(G,H,c)$ of the problem class, and where MAX is  the maximum value of an integer multiflow, and  $f$ is a  multiflow for this same input, and
the {\em half-integrality gap} is defined similarly by replacing ``integer'' by ``half-integer'' in the nominator.  A $\rho$-approximation algorithm $(\rho\in\mathbb{R})$ for a maximisation problem is a polynomial algorithm which outputs a solution of value at least $\rho$ times the optimum; $\rho$ is also called the {\em approximation ratio} (or guarantee). All these gaps and ratios are defined analogously for minimisation problems.

Our first result (Section~\ref{sec:multicuts}, Theorem~\ref{maxflow mincut}) is an upper bound of $2$ for the flow-cut gap (i.e. multicut/multiflow ratio) for plane instances, the missing relation we mentioned.  We prove this by relating multicuts to $2$-edge-connectivity-augmentation  in the plane-dual, and applying a bound of  Williamson, Goemans, Mihail and Vazirani \cite{williamson1995primal} for this problem.

We next show (Section~\ref{sec:half}, Theorem~\ref{thm:half}) how to obtain a half-integer flow from a given (fractional) flow in plane instances, reducing the problem to a linear program with a particular combinatorial structure, and solving it. 

Finally (Section~\ref{sec:int}, Theorem~\ref{thm:int}), given any feasible half-integer flow, we show how to extract an integer flow of  value at least half of the original, in polynomial time, using an algorithm that $4$-colors planar graphs efficiently \cite{robertson1996efficiently}. 

These results imply an integrality gap of  $1/2$ for maximum half-integer flows, $1/4$ for maximum integer flows, and the same approximation ratios for each. 
In turn, the flow-cut gap of $2$ implies a half-integer flow-cut gap of $4$ and an integer flow-cut gap of $8$ for plane instances. 

Section~\ref{sec:examples} presents some ecamples providing lower bounds for multicuts and upper bounds for multiflows. A summary of the results completed by lower bounds and open problems are stated in Section~\ref{sec:conclusions}. The Appendix  is a guide to some of the tools.

These are the first constant approximation ratios and gaps proved for the studied problems.\footnote{A shortened version of this manuscript has been submitted to IPCO21 in November 2019 and will appear in the volume of the conference (LNCS, Springer) in June 2020.}  Approximately at the same time as our submission to IPCO21, a  manuscript was published on Arxiv, proving
constant ratios for the same problems \cite{paris2020an}. The bounds for the gaps and the approximation ratio of \cite{paris2020an} are less sharp, but an interesting new rounding method, and a new complexity result appear in it. The latter will occur in the following section.



%% file: preliminaries_iv.tex
\section{Preliminaries}\label{sec:preliminaries}

We detail here some notions, notations, terminology, facts and tools we use, including some preceding results of influence.  


\medskip\noindent
{\bf Demands, the Cut Condition and Plane Duality} 

We describe the notion of demand flows, which is closely related to multiflows defined in the previous section. The problem is defined by the quadruple $(G,H,c,d)$, where $G,H,c$ are as before, {\em demands}  $d:F\rightarrow\mathbb{Z}_+$ are given, and we are looking for a feasible (sometimes in addition integer or half-integer) flow $f$ satisfying $\sum_{\{P\in\mcP:e\in P\}}=d(e)$ for all $e\in F$. 



A {\em cut}  in a graph $G=(V,E)$ is a set of edges of the form 
$\delta(S)=\delta_{E}(S):=\{e\in E: e \hbox{ has exactly one endpoint in $S$}\}$ for all $S\subseteq V$. Note that $\delta(S)=\delta(V\setminus S)$;  $S$, $V\setminus S$ are called the {\em shores} of the cut. For a subset $E'\subseteq E$ we use the notation $c(E'):=\sum_{e\in E'} c(e)$, and we adopt this usual way of extending a function on single elements to subsets. For instance, $d(F'):=\sum_{e\in F'}d(e)$ is the  {\em demand} of the set $F'\subset F$. 

A necessary condition for the existence of a feasible multiflow is the so called {\em Cut Condition}: for every $S \subseteq V, c(\delta_{E}(S)) \geq d(\delta_{F}(S))$, that is,  the  capacity of each cut must be at least as large as its demand. The cut condition is not sufficient for a flow in general, but Seymour \cite{seymour1981odd} showed that it is sufficient for an integer flow,  provided  $G+H$ is planar, the capacities and demands are integer, and their sum is even on the edges incident to any vertex. A half-integer flow follows then for arbitrary integer capacities for such {\em plane} instances (see Section~\ref{sec:introduction}).
The same holds for Okamura-Seymour instances \cite{okamura1981multicommodity}.
There are more examples, unrelated to planarity, where the cut condition is sufficient to satisfy all demands, and with an integer flow, for instance when all demand edges can be covered by at most two vertices.

Seymour's theorem \cite{seymour1981odd} on the sufficiency of the cut condition and about the existence of integer flows has promoted plane multiflow problems to become one of the targets of investigations. {\bf This paper is devoted to plane multiflow maximisation.} Seymour's proof is based on a nice correspondence to other combinatorial problems through plane duality, that we also adopt. 

Middendorf and Pfeiffer \cite{middendorf1993complexity} showed that the plane multiflow problem ( already deciding the existence of edge disjoint paths) is NP-Hard, which immediately  implies that {\em plane multiflow maximisation is also NP-hard} (see Footnote~\ref{foot:hard}). The cut condition can be checked in polynomial time\footnote{Seymour's correspondence through dualisation reduces this problem to checking whether $F^*$ is a minimum weight ``$T_{F^*}$-join'' (see eg. \cite{schrijver2003combinatorial}) in $(G+H)^*$ with weights defined by $c$ and $d$, where $T_{F^*}$ is the set of odd degree vertices of $F^*$. This also provides a polynomial separation algorithm for maximising the sum of (not necessarily integer) demands satisfying the cut condition.\label{foot:cutcondition}} so the difficult question to decide is the existence or not of an integer flow when the cut condition is satisfied. As far as multicuts are concerned, we  show that their minimisation in plane instances is equivalent to a 2-edge-connectivity augmentation problem in planar graphs.

Following Schrijver \cite[p. 27]{schrijver2003combinatorial}  we denote the dual of the planar graph $G$ by $G^*=(V^*,E^*)$, where $V^*$ is the set of faces of $G$ and each edge $e\in E$ corresponds to an edge $e^*\in E^*$ joining the two faces that share $e$. For $X\subseteq E$ denote $X^*:=\{e^*:e\in X\}$. An important fact we need and use about dualisation: {\em $C$ is an inclusionwise minimal cut in $G$ if and only if $C^*$ is a circuit in $G^{*}$.} 

It is tempting to reduce multiflow maximisation to demand flows, using the sufficiency of the cut condition. The plane dual of the question is to maximise the sum of the absolute values of weights (demands) on the negative edges so that there is no negative cycle in the graph. An integer solution to this problem does not imply an integer multiflow, cf. $(C_4,2K_2)$\footnote{This notation means $G=C_4$ and $H=2K_2$, where $C_4=(V,E)$ is the graph consisting of a circuit on $V$,  $|V|=4$, and $2K_2=(V,F)$ is its complement, so $G+H=K_4$. The capacities and the demands are defined to be $1$. To transform this example to multiflow maximisation with integrality gap $1/2$ we keep the capacities, and wish to put an upper bound on the demands, otherwise there is an integer flow of value $2$. Such an upper bound can be realized by replacing each demand edge by two edges in series, a demand and a supply edge, the latter of capacity $1$, getting from  $(C_4,2K_2)$ the MEDP for $(\overline{C_4},\overline{2K_2})$, with integrality gap $1/2$.  This is in general a useful transformation from demand problems to maximisation, for which we will use the same ``overline'' notation. This is at the same time a reduction of  the NP-hardness of plane multiflow feasibility, known to be NP-hard  \cite{middendorf1993complexity} to plane multiflow maximisation.\label{foot:K4}}. The complexity of finding such an integer solution was open,  has been proved recently in \cite[Theorem 2]{solal2019exercices} to be NP-hard when $G+H$ is planar after deleting a vertex,  \cite{paris2020an} proves this  when $G+H$ is planar, but it remains open when $H$ consists of one tree and isolated vertices; then, in addition, \cite{korachpenn1992gap} provides an integer solution whenever the cut condition is satisfied.

\medskip\noindent
{\bf Two Lemmas on Laminar Families}

This  correspondence by plane duality allows to transform any fact on cuts to circuits in the dual and vice versa. For example fractional, half-integer or integer packings of cuts in $G^*$, where each cut contains exactly one edge of $F^*$ correspond to a fractional, half-integer or integer multiflow in $G$. We provide now some related definitions, notations  and facts. We do this directly on the graph $(G+H)^*$ where they will be used; we denote by $V^*$ the vertices of this graph, that is the faces of $G+H$. So  $(G+H)^*=(V^*,E^*\cup F^*)$.  

Let  $\delta (A)=\delta_{E^*\cup F^*}(A)$, $\delta (B)=\delta_{E^*\cup F^*} (B)$ $(A, B\subseteq V^*)$ be two {\em crossing} cuts, i.e. $A$ and $B$ are neither disjoint nor contain one another, each of which contains exactly one edge of $F^*$. They can be replaced by $\delta (A^*\cap B^*)$ and $\delta(A^*\cup B^*)$  or $\delta (A^*\setminus B^*)$ and $\delta (B^*\setminus A^*)$(in the plane dual). It is easy to check that every edge is contained in at most as many cuts after the replacement as before, and if both cuts contain exactly one edge of $F^*$ then this also holds for the replacing cuts. Doing this iteratively and using plane duality, we can convert any feasible flow into another one (without changing the value of the flow) in which no two flow paths cross. Lemma~\ref{lem:lamexists} below formalises exactly what we need.


A family of subsets of $V^*$ is called {\em laminar} if any two of its members are disjoint or one of them contains the other. If for any two members one of them contains the other we say that the family is a {\em chain}. Given a laminar family $\mcL\subseteq 2^{V^*}$, a chain $\mcC\subseteq \mcL$ is {\em full} (in $\mcL$) if $X, Y, Z\in \mcL, X\subseteq Y\subseteq Z$ and  $X, Z\in\mcC$ implies  $Y\in\mcC$. We call a multiflow $f$  {\em laminar}, if $\{C^*: C=P\cup \{e_P\},  P\in\mcP, f(P)>0\}=\{\delta_{E^*\cup F^*}(L): L\in \mcL\}$, where $\mcL\subseteq 2^{V^*}$ is  laminar. We state the following well-known Lemma without proof (see the Appendix for references and explanations): 

\begin{lemma}\label{lem:lamexists} 
	For every feasible multiflow $f$ there  exists  a  laminar feasible multiflow $f'$ so that $|f'|=|f|$, and $f'$ can be found in polynomial time.  
\end{lemma}

The following useful properties are easy to check: 

For a family $\mcL$ of subsets of $V^*$ and $a\ne b\in V^*$, denote $\mcL(a):=\{L\in\mcL:  a\in L\}$; $\mcL(a,b):=\{L\in\mcL:  a\in L,\, b\notin L\}$. 

\begin{lemma}\label{lem:lampprop} Let $\mcL$ be a laminar family of subsets of   $V^*$.  Then 
	
	a. $|\mcL|\le 2(|V^*| -1)= O(|V|)$. 
	
	b.  $\mcL(a,b)\subseteq \mcL(a)\subseteq V^*$ both form full chains of subsets of $V^*$. \qed 
	 
\end{lemma}



\medskip\noindent
{\bf Integrality in demand flows and stable sets}  


We first compute a half integer flow of value at least half the fractional flow and then convert this into an integer flow. We now describe an instance which illustrates the difficulty in finding an integer flow. Consider a planar graph with a perfect matching without any nontrivial {\em tight cut}\footnote{i.e. a cut with both shores containing more than one vertex, and meeting every perfect matching in exactly one edge. Lov\'asz characterised graphs without nontrivial tight cuts as ``bicritical $3$-connected graphs'' \cite{LovaszPlummer1986MatchingTheory} called {\em bricks}. Such graphs may have arbitrarily many vertices, even under the planarity constraint.}  to be $(G+H)^*$,  and $F^*$ to be any perfect matching in it. Let all capacities be $1$. In order to have an upper bound  of $1$ on each demand, apply the transformation of Footnote~\ref{foot:K4}. Then multiflows can use only the dual edge-sets of stars of vertices in $(G+H)^*$, so an  integer multiflow of value $k$ corresponds exactly to a {\em stable set}, that is, a set of vertices not inducing any edge, of size $k$ in  $(G+H)^*$. This indicates that in order to find a large integer flow, we sometimes need to find large stable sets in planar graphs. 

Despite the restriction of the paths to those arising from duals of stars of vertices, the gap between the integer and the half-integer max-flow is at least $1/2$ for these bricks: it follows from the $4$-color theorem  \cite{AppelHaken1976proof} that a stable set of size at least $|V^*|/4$  exists while the maximum half integer flow cannot exceed $|F|=|V^*|/2$. We will be able to reach this ratio in general (see Section~\ref{sec:int}, Theorem~\ref{thm:int}), and $(\overline{C_4},\overline{2K_2})$, defined as the graph in Footnote~\ref{foot:K4} but instead of   shows that this cannot be improved.      


In order to reach this integrality gap of $2$ in general,  we will actually need to find a stable-set of size $n/4$ in a planar auxiliary graph of $n$ vertices. The maximum stable set problem is NP-hard, but there is a PTAS for it in planar graphs \cite{Baker1994}, which, combined with the $4$-color theorem \cite{AppelHaken1976proof} provides a stable-set of size $n/4$;  an alternative is to use the $4$-coloring algorithm of Robertson, Sanders, Seymour and Thomas \cite{robertson1996efficiently} which directly provides a $4$-coloring of a planar graph in polynomial time, and the largest color class is clearly of size at least  $n/4$. Either of these algorithms can be used as a black-box-tool for rounding half-integer flows, so we state the result: 

\begin{lemma}\label{lem:stable} 
	In a planar graph on $n$ vertices, a stable set of size $n/4$ can be found in polynomial time.
\end{lemma}

A similar rounding argument appeared in the work of  Fiorini, Hardy, Reed and Vetta  \cite{fiorini2007odd} in the somewhat different context of proving an upper bound of Kr\'al and Voss \cite{KralVoss2004odd} for the ratio between ``minimum size of an odd cycle edge transversal'' versus the ``maximum odd cycle edge packing''   using the $4$-color theorem \cite{AppelHaken1976proof}. Our procedure in Section~\ref{sec:int}, occurs to be more general in that it is starting from an {\em arbitrary, not necessary optimal,  half-integer multiflow} for an arbitrary capacity function, and constructs an integer multiflow, with a polynomial algorithm.  Similar difficulties   approached independently with the $4$-color theorem confirm that it may be unavoidable for  plane multiflows.



%% file: multicut-2-connectivity_iv.tex
\section{Multicuts versus Multiflows via 2-Connectors}\label{sec:multicuts}

We show in this section that the flow-cut gap is at most two for plane instances, via a reduction to the 2-edge-connectivity augmentation problem. 

Given $G=(V,E)$,  $H=(V,F)$, a {\em $2$-connector for $H$ in $G$} is a set of edges  $Q \subseteq E$ such that none of the edges $e\in F$ is a cut edge of $(V,Q \cup F)$; equivalently, $Q$ is a $2$-connector if and only if each $e\in F$ is contained in a circuit of $Q \cup F$.
The {\em 2-edge-connectivity Augmentation Problem (2ECAP)} is to find, for given edge costs $c:E\rightarrow \mathbb{Z}_+$ on $E$, a minimum cost $2$-connector.


Let $(G,H,c)$, $G=(V,E)$, $H=(V,F)$ be the input of a plane multiflow maximisation  problem, and   $(G+H)^*=(V^{*},E^{*} \cup F^{*})$, where $V^*$ is the set of faces of $G+H$. Define $c(e^*):=c(e)$ $(e\in E)$.  
\begin{lemma}\label{lem:multicutconnector}
	The edge-set $Q\subseteq E$ is a multicut for  $(G,H)$ if and only if $Q^{*}$ is a $2$-connector  for $(V^{*}, F^{*})$ in $(V^{*},E^{*})$.
\end{lemma}
\begin{proof}
	The edge-set $Q \subseteq E$ forms a multicut in $G$ if and only if the endpoints $u$, $v$  of any edge $e\in F$ are in different components  of  $(V,E \setminus Q)$, that is, if and only if for all $e\in F$ there exists an inclusionwise minimal cut $C\subseteq Q\cup F$  of  $G+H$ such that $e\in C$. But we saw among the preliminaries concerning duality that $C$ is an inclusionwise minimal cut in $G+H$ if and only if $C^*$ is a circuit in $Q^*\cup F^*$. Summarising,  $Q \subseteq E$ forms a multicut in $G$, if and only if for all $e^*\in F^*$ there exists a circuit  $C^*$ in $Q^*\cup F^*$ such that $e^*\in C^*$.   This means exactly that $Q^*\subseteq E^*$ is a $2$-connector for $(V^{*}, F^{*})$, in $(V^{*},E^{*})$ , finishing the proof. 
	\qed
\end{proof}


Let $(G,H,c)$ be an instance of multiflow problem with $G+H$ planar. Let $p: 2^{V^*} \rightarrow \{0,1\}$ with $p(S)=1$ if and only if  $|\delta (S)\cap F^*|=1$, otherwise $p(S)=0$ $(S \subseteq V^{*})$. The following linear program specialises the one investigated in 
\cite{williamson1995primal}:  
\begin{lp}{minimise}{\sum_{e \in E^{*}} c(e) x(e),}
	\cnstr{\sum_{e \in \delta(S) } x(e)}{\ge}{p(S),}{ S \subseteq V^{*};\qquad\qquad (2ECAP)}
	\cnstr{x(e)}\ge 0 {e \in E^{*}.}
\end{lp}
Since $p$ is $\{0,1\}$-valued so are the coordinatewise minimal integer solutions including the integer optima of (2ECAP).   
{\em The $\{0,1\}$-solutions are exactly the (incidence vectors of) $2$-connectors of $(V^*, F^*)$ in $(V^*, E^*)$.} In the dual linear program of (2ECAP), we have a variable $y(S)$ for all $S \subseteq V^{*}$, constraints $\sum_{S: e \in \delta(S) } y(S) \le c_{e}$ for all ${ e \in E^{*}}$ and $y(S) \ge 0$ for all $S \subseteq V^{*}$. The objective is to maximise $\sum_{S \subseteq V^{*}} p(S)y(S)$. Williamson, Goemans, Mihail and  Vazirani \cite{williamson1995primal} developed a primal-dual algorithm finding for given input $(G+H)^*$ and $c$, an {\em  integer primal solution $x_{\small\rm WGMV}$ to a class of linear programs including (2ECAP), together with a (not necessarily integer) dual solution $y_{\small\rm WGMV},$ in polynomial time,} proving the following {\em WGMV inequality} \cite[Lemma 2.1]{williamson1995primal}, see also \cite[Section 20.4]{kortevygen2018combinatorial}: 
\[{\rm LIN}\le OPT\le \sum_{e \in E^{*}} c(e) x_{\small\rm WGMV}(e)\le 2 \sum_{S\subseteq V^*, p(S)=1} y_{\small\rm WGMV}(S)\le 2\,{\rm LIN}\le 2\,{\rm OPT},\]
where  OPT is the minimum cost of a $2$-connector, and LIN is the optimum of (2ECAP). We will refer to this algorithm as the {\em WGMV algorithm}.

Note that the $WGMV$ algorithm works for the class of weakly supermodular functions. A function $h :2^{V} \rightarrow \{0,1\}$ is called {\em weakly supermodular} if $h(V)=0$ and for any $A,B \subseteq V$, $h(A)+h(B) \leq \max \{h(A \cap B)+h(A \cup B),h(A \setminus B) + h(B \setminus A)\}$. It can be verified that $p$ defined above is weakly supermodular.




\begin{theorem}\label{maxflow mincut}
	Let $(G, H, c)$ be a plane multiflow problem. Then there  exists a feasible multiflow $f$ and a multicut $Q$, such that $c(Q)\le 2|f|$, where both $f$ and $Q$  can be computed in polynomial time.
\end{theorem} 

\begin{proof}  Recall that the WGMV algorithm finds $x_{\small\rm WGMV}$ and $y_{\small\rm WGMV}$ satisfying the WGMV inequality, where $x_{\small\rm WGMV}$ is the incidence vector of a $2$-connector of $(V^*, F^*)$ in $(V^{*},E^{*})$ , let us denote its plane  dual set in $G+H$ by $Q_{\small\rm WGMV}$.  By Lemma~\ref{lem:multicutconnector} $Q_{\small\rm WGMV}$ is a multicut. 
	
Consider a set $S$ with $p(S)=1$, that is, $|\delta(S)\cap F^*|=|\delta(S)^*\cap F|=1$, all the other sets can be supposed to be absent from the inequalities. 
Then  $\delta(S)^*$  contains a circuit  $C$  of  $G+H$ containing the unique edge  of $|\delta(S)^*\cap F|$. Therefore $C\setminus F$ is a path in $G$, denote it by $P_S$.  Define a multiflow $f$ in $G+H$  by $f(P_S)=y_{\small\rm WGMV}(S)$. The dual feasibility of $y_{\small\rm WGMV}$ means exactly that the multiflow $f$ is feasible.  By our construction and the WGMV inequality we have 
	\[c(Q_{\small\rm WGMV})=  \sum_{e \in E^{*}} c(e) x_{\small\rm WGMV}(e)\le 2 \sum_{S\subseteq V^*, p(S)=1} y_{\small\rm WGMV}(S) = 2|f|.\]
	So the multicut $Q:=Q_{\small\rm WGMV}$ and the multiflow $f$ satisfy the claimed inequality; all operations, including the WGMV algorithm run  in polynomial time. \qed \end{proof}

Note that if $y_{\small\rm WGMV}$ is half-integer (assuming integer edge-costs), the obtained multiflow is half-integer and a half-integer flow-cut gap of $2$ directly follows. There are examples where the WGMV algorithm does not produce a half-integer dual solution, but we do not know of an instance where half-integer flow-cut gap is larger than 2.  

%% file: frac_half_iv.tex
\section{From Fractional to Half-Integer}\label{sec:half}

We show here how to convert a flow to a half-integer one, loosing at most half of the flow value, provided the capacities are integers.

\begin{theorem}\label{thm:half} Let $(G, H, c)$ be a plane multiflow problem, where  $c:E\rightarrow\mathbb{Z}_+$ is an integer capacity function.  Given a feasible multiflow $f$, there exists a  feasible half-integer multiflow $f'$, such that $|f'|\ge |f|/2$, and it can be computed  in polynomial time.
\end{theorem}

\begin{proof} By Lemma~\ref{lem:lamexists} we can suppose that the given feasible multiflow $f$ is laminar and can be found in polynomial time. Let  $\mcL\subseteq 2^{V^*}$  be the laminar family of cuts in   $(G+H)^*$,    with
	$\{C^*: C=P\cup\{e_P\}, f(P)> 0\}=\{\delta_{E^*\cup F^*}(L): L\in \mcL\}$ (see Section~\ref{sec:preliminaries}, just above Lemma~\ref{lem:lamexists}). Denote $f_L:=f(\delta (L)^*\setminus F)$, $(L\in\mcL)$. The feasibility of the multiflow $f$ means 
	$\sum_{L\in\mcL, e\in \delta_{E^*}(L)}f_L\le c(e)$, that is, $x_L=f_L\in\mathbb{N}_+$ satisfies
	\begin{equation}\label{eq:flow}
	\sum_{L\in\mcL}x_L = |f|, \hbox{and}\quad \sum_{L\in\mcL, e\in \delta_{E^*}(L)}x_L\le c(e),  \hbox{for all $e\in E$, $x_L\ge 0,\, (L\in\mcL)$.} 
	\end{equation}
	Clearly, the edge $e=uv$ is contained in exactly those sets $\delta(L)$ $(L\in\mcL)$ for which $L\in L(u,v)\cup L(v,u)$, where $\mcL(a,b):=\{L\in\mcL:  a\in L,\, b\notin L\}$ $(a,b\in V^*)$, and both $\mcL(u,v)$ and $\mcL(v,u)$ form full chains (Lemma~\ref{lem:lampprop}). So the linear program
	\begin{equation}\label{eq:relax}
	\begin{aligned}
	& \underset{}{\text{max}}
	& & \sum_{L\in\mcL} x_L \\
	& \text{subject to} & &  \sum_{L\in\mcL(u,v)} x_L \leq c(e)  & \text{for all }  (u,v)\in V^*\times V^*, uv=e^*\in E^*\\
	& & &  x_L \geq 0 \; & \text{for all } L\in\mcL,\\
	\end{aligned}
	\end{equation}
	is a relaxation of (\ref{eq:flow}): for each  $u, v\in V^*, uv=e^*\in E^*$ the coefficient vector of (\ref{eq:flow}) associated with $e^*$ is the sum of the   coefficient vectors, one  for each of  $(u,v)$ and $(v,u)$, of the two inequalities associeated with $e^*$ in (\ref{eq:relax}).  Both of these ($\mcL(u,v)$ and $\mcL(v,u)$) correspond  to full chains in the laminar family $\mcL$, and the right  hand side $c(e)$ is repeated for both. 
	
	Denote $M=M(f)$ the $2|E|\times |\mcL|$ coefficient matrix of (\ref{eq:relax})  (without the non-negativity constraints).

	According to   Edmonds and Giles  \cite{edmondsgiles1977}, $\mcL$ has a rooted tree (arborescence) representation in which the full chains correspond to subpaths of paths from the root, so $M$ is a network matrix. As such, it is totally unimodular by Tutte \cite{tutte1965networkunimodular} and (\ref{eq:relax})  has an integer optimum $x$, computable in polynomial time by \cite{hoffmankruskal1956integral}, \cite[Theorem 19.3 (ii), p. 269]{schrijver1986LPandIP}.  
	
	To finish the proof now, note that on the one hand,  $f_L$ $(L\in\mcL)$ is a solution to (\ref{eq:flow}), and therefore it is also a feasible solution of the relaxation  (\ref{eq:relax}). Since $x_L$ $(L\in\mcL)$ is the maximum of (\ref{eq:relax}),  $\sum_{{L} \in \mcL}  x_L \ge  \sum_{{L} \in\mcL} f_L=|f|$.  According to the two inequalities of (\ref{eq:relax}) associated to $e^*$, the sum of coefficients of the paths containing any given edge $e\in E$ is at most $c(e)+c(e)=2\, c(e)$, so $f'=x/2$ defines a half-integer feasible flow in $(G,H,c)$ (by assigning the flow value $f'(L)$ to the path $\delta(L)^*\setminus F$,  so  $|f'|\ge |f|/2$, finishing the proof.  
	\qed
	\end{proof}

For more explanations and references  showing that $M$ is a network matrix, and an alternative direct combinatorial solution of the integer linear program with a simple  greedy-type algorithm, see the Appendix. 

This proof does not fully exploit the possibilities of totally unimodular matrices: instead of putting $c(e)$ as right hand side for both inequalities of (\ref{eq:relax}) associated with $e^*$ we can put everywhere the smallest integer capacities satisfied by the fractional flow. Since the fractional values on the mentioned two inequalities sum to at most $c(e)$ and not $2c(e)$ we get a sharper result this way. The proof works if we replace the capacities by the rounded up fractional flow, leading to an error of only $1$ compared to the original capacities, due to the round-up. Let us denote by $\underline 1$ the all $1$ function on $E$, and check this precisely:  

\begin{theorem}\label{thm:c+1} Let $(G, H, c)$ be a plane multiflow problem, where  $c:E\rightarrow\mathbb{Z}_+$.  Given a feasible multiflow $f$, there exists a feasible integer multiflow $f'$, computable in polynomial time, feasible for the capacity function $c+\underline 1$, and $|f'|\ge |f|$.
\end{theorem}

\begin{proof} Let $f\in\mathbb{R}^{\mcL}$ be a feasible multiflow, and
$M=M(f)$ the $2|E|\times |\mcL |$ coefficient matrix defined in the proof of Theorem~\ref{thm:half}. Define $d:=Mf\in\mathbb{R}^{2|E|}$, and consider the linear program 
	\begin{equation}\label{eq:c+1}
\begin{aligned}
& \underset{}{\text{max}}
& & \sum_{L\in\mcL} x_L \\
& \text{subject to} & &  \sum_{L\in\mcL(u,v)} x_L \leq\lceil d(u,v)\rceil  &\text{for all }  (u,v)\in V^*\times V^*, uv=e^*\in E^*\\
& & &  x_L \geq 0 \; & \text{for all } L\in\mcL,\\
\end{aligned}
\end{equation}
where   $d(u,v):=\sum_{L\in\mcL(u,v)}f_{\delta(L)^*\setminus F}$. In words, (\ref{eq:c+1}) has exactly the same coefficients as (\ref{eq:relax}), 
but the right hand sides $d(u,v)$ and $d(v,u)$ of the two inequalities associated with $e\in E$, $e^*=uv$ are defined with the sum of flow values on $\delta(L)^*\setminus F$,  for $L\in\mcL(u,v)$ and  $L\in\mcL(v,u)$ respectively. 

Since $f$ is a feasible flow for $(G,H,c)$,   $d(u,v) + d(v,u)\le c(e)$, so  $\lceil d(u,v)\rceil+\lceil d(v,u)\rceil\le c(e)+1$, and $f$ is feasible for (\ref{eq:c+1}) since the capacities have been defined by rounding up the flow values. On the other hand, the coefficient matrix is totally unimodular, so by \cite{hoffmankruskal1956integral}, \cite[Theorem 19.3 (ii), p. 269]{schrijver1986LPandIP} the linear program (\ref{eq:c+1}) has an integer maximum solution $f'$, again  computable in polynomial time, and $f'\le c+\underline 1$, $|f'|\ge |f|$.
\qed
\end{proof}

Theorem~\ref{thm:half} is an immediate consequence of Theorem~\ref{thm:c+1}:

For each $k\in\mathbb{N}$, $k\ge 1$,  $\frac{k+1}{2}\le k$ holds, so (after deleting $0$ capacity edges)  $(c + \underline 1)/2  \le  c$, and therefore, dividing by $2$ the primal optimum of (\ref{eq:c+1}), we immediately get a half-integer solution to (\ref{eq:relax}). 


This result extends a natural consequence for maximisation of the tight additive integrality gap known for plane demand flow problems with integer capacities and demands satisfying the cut condition: according to a result of Korach and Penn \cite{korachpenn1992gap}, if all demand edges lie in two faces of the supply graph, {\em there exists an integer multiflow satisfying all demands but at most $1$.} This readily implies that increasing each capacity by $1$,  an integer  flow of the same value as the maximum flow for the original capacities, can be reached.

From Frank and Szigeti \cite{frankszigeti1995surficit}  the same can be deduced  for demand-edges lying on an  arbitrary  number $k$ of faces. Indeed, increasing all capacities by $1$, {\em  the surplus of the cut condition will be at least $k$, which is the Frank-Szigeti condition for integer multiflows. } Theorem~\ref{thm:c+1} states that this consequence is actually true in general, without requiring the integrality of the demands, and also for the  maximisation problem; the same holds for maximum ``packings of $T$-cuts''. 

%% file: integralflow_iv.tex
\section{From Half-Integer to Integer}\label{sec:int}
In this section, we show how to round a half-integer flow to an integer one, losing at most one half of the flow value.

\begin{theorem} \label{thm:int}
	Let $(G, H, c)$ be a plane multiflow problem, where $G=(V,E)$,  $H=(V,F)$ and $c:E\rightarrow\mathbb{Z}_+$.  Given a feasible half-integer multiflow $f$, there exists a feasible integer   multiflow $f'$, computable in polynomial time, and $|f'|\ge |f|/2$.
\end{theorem}

\begin{proof} Let $(G,H,c)$ and  $f$ be as assumed in the condition, moreover that $f$ is laminar (Lemma~\ref{lem:lamexists}). We proceed by induction on the integer $2|f|$. We suppose that {\em  all nonzero values $f(P)>0$ $(P\in\mcP)$ are actually $1/2$}: if $f(P)\ge 1$, we can decrease $f(P)$ by $\lfloor f(P)\rfloor$, as well as all capacities of edges of $P$, and the  statement follows from the induction hypothesis. (Such a step can be repeated only a polynomial number of times, since $|\mcP|=O(|V|)$ by Lemma~\ref{lem:lampprop}.)
	
To choose the values to round we replace $c(e)$, by $c(e)$ parallel edges of capacity $1$ each.\footnote{This is not an allowed step for a polynomial algorithm, but it will not really be necessary to do it. The choice of the cuts  to be rounded down or up  will be clear from the proof without actually executing this subdivision.  The choices for rounding concern a family of size  $O(|V|)$ .} Then take the plane dual of the resulting graph, which is  $(G+H)^*$ with each edge $e^*\in E^*$   replaced by a path of size $c(e)$. We consider the laminar system $\mcL$ defining the paths $P\in\mcP$, $f(P)>0$ \footnote{$\{\delta_{E^*\cup F^*}(L)^*:L\in\mcL\}=\{P\cup e_P: P\in\mcP, f(P)>0\}$,  as before.} in this subdivided graph so that every edge is contained in at most two sets $L\in\mcL$, and remains laminar (this is clearly possible, since all positive $f_P$ values are $1/2$). For simplicity, we keep the notations of the original graph - as if what we get in this way were the given graph with all capacities equal to $1$.  

Let $I(\mcL)=(\mcL,M)$ be the intersection graph of the cuts defined by $\mcL$, that is, $M:=\{L_1L_2: L_1,L_2\in\mcL, \delta(L_1)\cap \delta(L_2)\ne\emptyset\}$. We have Claim~a. and b. so far, and now we check Claim~c.:

\medskip\noindent{\bf Claim}: a. Each $e\in E^*$ is contained in at most two sets in $\{\delta(L):L\in\mcL\}$. 

b.  $|f|=|\mcL|/2$.

c. $I(\mcL)$ is planar. 

\smallskip
To check Claim~c., note first its validity if $\mcL$ consists of disjoint sets. If this does not hold, not even by {\em complementing}  some $L\in\mcL$ (i.e. replacing it by $V^*\setminus L$), then it is easy to find (possibly after coplementation) three sets $L_1\subset L_0\subset L_2$ in $\mcL$. We claim that $L_0$ is a cut-vertex in  $I(\mcL)$. By laminarity, every cut $\delta(L)$ $(L\in\mcL)$ has either a shore $A$ contained in $L_0$ (like $L_1$), or a shore $B$ disjoint from $L_0$ (like  $V\setminus L_2$). If there exists an $e\in\delta(A)\cap\delta(B)$, this would mean that $e$ has an endpoint in  $A\subseteq L_0$ and the other endpoint in $B\subseteq V\setminus L_0$, so $e\in\delta(L_0)$, contradicting Claim~a. So {\em $L_0$ is a cut vertex, i.e.
$\mcL=\mcL_1\cup\mcL_2$, $\mcL_1\cap\mcL_2=\{L_0\}$, with no edge between $\mcL_1$ and $\mcL_2$ in $I(\mcL)$.}

 Since the graphs induced by $\mcL_i$ ($i=1,2$) are both defined by flows of smaller value, we can apply the induction hypothesis to them: they are planar, so $I(\mcL)$ is also planar, and the Claim follows.  

\smallskip To finish the proof of the theorem using Claim~c., find a stable set of size $|\mcL|/4$ in $\mcL$, by Lemma~\ref{lem:stable}, and increase the flow on the corresponding paths $\delta(L)^*\setminus F$ to $1$, while decreasing the flow on the other paths to $0$. This results in a feasible integer flow $|f'|\ge |\mcL|/4=|f|/2$, finishing the proof of the bound.

	
	
	\smallskip
	For the computational complexity results, first recall that the support of the half integer laminar flow $f$ obeys Lemma~\ref{lem:lampprop}a.  Then note, that among this linear number of sets, the proof finds in polynomial time at least one fourth of the cuts to round up, while the other cuts are rounded down, so that the capacity constraints are not violated. It is straightforward to mimick the subdivision of edges without doing it, and to compute an input to  Lemma~\ref{lem:stable}, in strongly polynomial time. \qed
\end{proof}

%% file: integralitygap_iv.tex
\section{Examples}\label{sec:examples}
In the previous sections we showed upper bounds on   the flow-cut gap, lower bounds on  the (half)-integrality gap equal to the approximation guarantee of multiflow maximisation.   In this Section we present some examples proving lower bounds for multicuts and upper bounds for multiflows, showing the performance of the precedingly proved theorems.

A first simple example occurred already: $(\overline{C_4},\overline{2K_2})$ (Footnote~\ref{foot:K4}), where the maximisation of demands so that a feasible flow exists can be achieved in integers, but the maximum integer flow is only half of the maximum flow.  Another,  still simple and  small example shows already that a unique max-flow may not even be half-integer (Figure~\ref{fig:non-half-integer}).

\begin{figure}[t!]
	\centering
	\includegraphics[width=4.5in]{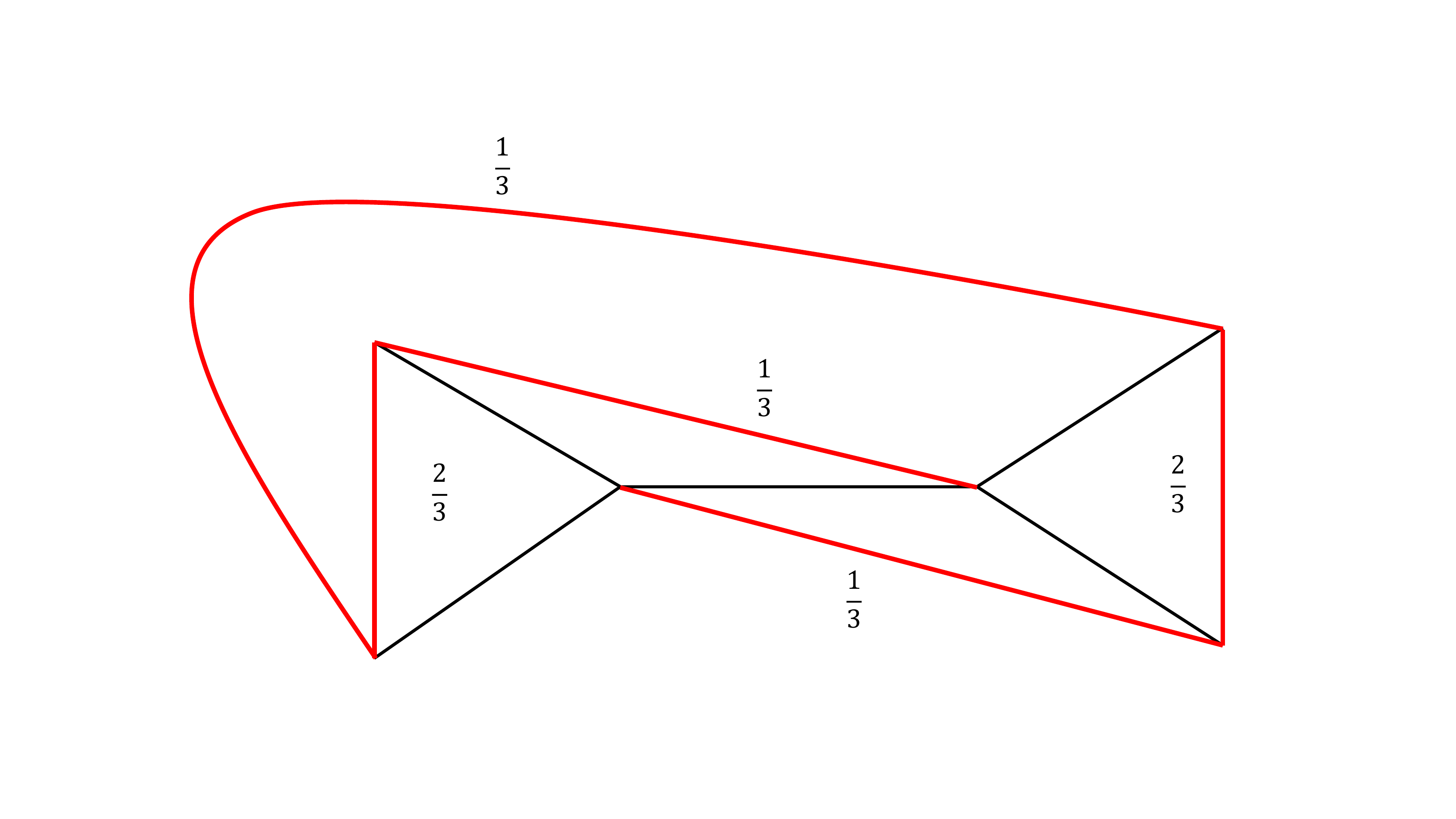}
	\caption{Black edges are supply while red (thick) ones are demand edges. All supply edges have capacity 1. The maximum flow value is $7/3$, while the value of the maximum half-integer (or integer) flow is $2$. The value of minimum multicut is $3$. So the (half-) integrality gap of this example is $6/7$, the flow-cut gap is at least $9/7$, the (half) integer flow-cut gap is at least $3/2$. }
	\label{fig:non-half-integer}
\end{figure}

The gap values of Figure~\ref{fig:non-half-integer} are not best possible. We provide now  an infinite class  of examples providing asymptotically the best possible values for some of the gaps.  For every non-negative integer $k$ we define a graph $G_k$, where for instance the half-integer flow-cut gap tends to $2$ as $k \rightarrow \infty$. (Cheriyan et.al. \cite{cheriyan2008integrality} used these instances to show integrality gap results for the Tree Augmentation Problem.)

Let $G_k=(V_k,E_k),\, H_k=(V_k,F_k),  k \geq 3$ be an instance of the multiflow problem defined as follows: $V_k=\{ a_{1},a_{2},\ldots,a_{k}\}$ $\cup \{b_{1},b_{2},\ldots,b_{k}\}, E_k=\{(a_{i},b_{i})| i \in [1,k]\} \cup \{(a_{i},a_{i+1})| i \in [1,k-1] \}$ and $F_k=\{(b_{i},b_{i+1})| i \in [1,k-1]\} \cup \{(b_{i},a_{i+2})| i \in [1,k-2]\}$. The capacity of all edges in $E_k$ is $1$ (see Fig.~\ref{fig:gap}).
\begin{figure}[htb] 
	\centering
	\includegraphics[width=3.5in]{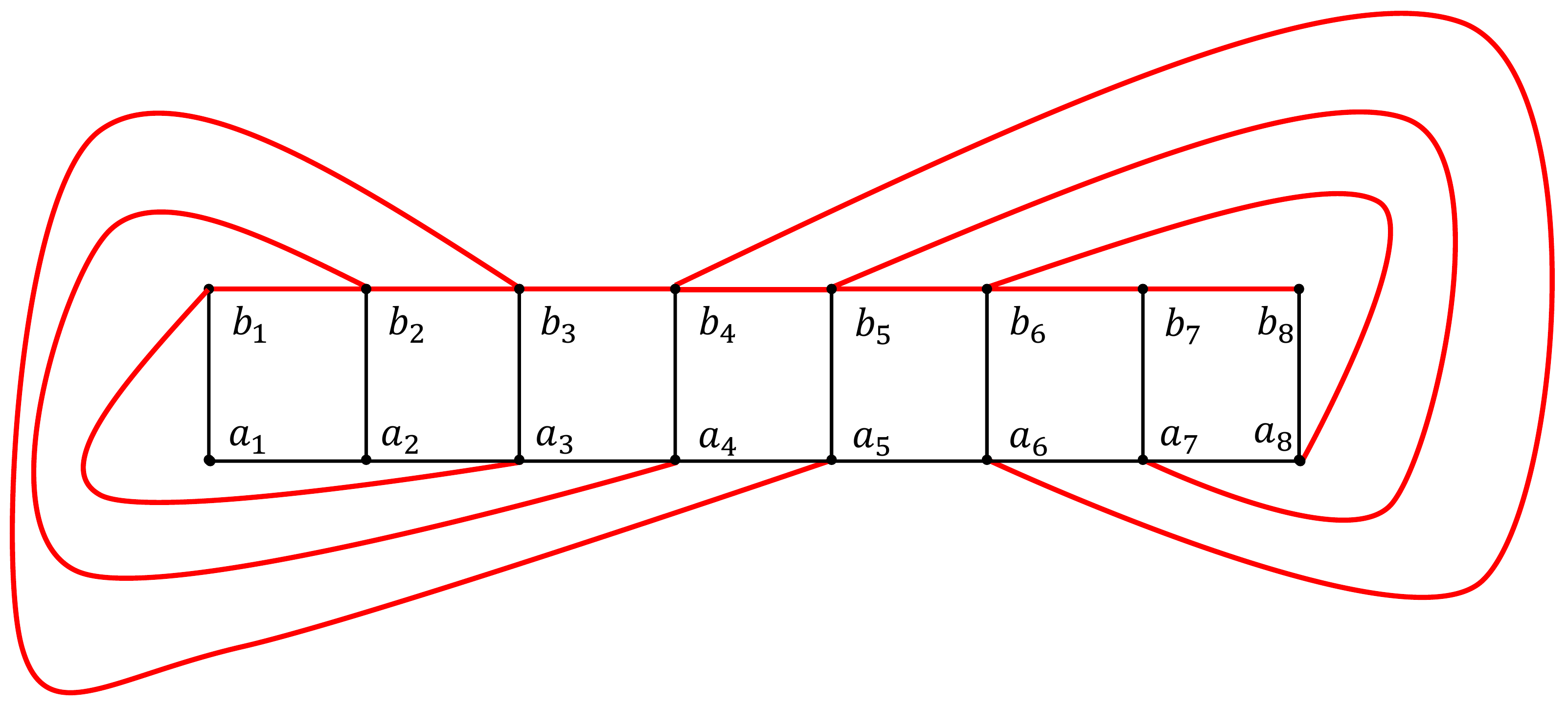}
	\caption{$G_8$: the capacity of supply edges is $1$; demand edges are red (thick).}
	\label{fig:gap}
\end{figure}
\begin{theorem} \label{theorem:gap}
The graph $G_k+H_k$ is planar for all $k=1, 2,\ldots$, and the following hold:
	
	-- The minimum multicut capacity is $k-1$.
	
	-- The maximum multiflow value is  $2(k-1)/3$.
	
	--  The maximum half-integer multiflow value is $k/2$.
	
	--  The maximum integer multiflow value is $\lfloor k/2\rfloor$.

\end{theorem}
\begin{proof} The minimum multicut capacity  is clearly at most $k-1$, since deleting each edge of the $(a_1,a_{k-1})$-path, the endpoints of all demand-edges are separated. We check now $|C|\ge k-1$ for an arbitrary multicut $C$, by induction. For $k=1,2$ the statement is obvious.  Deleting  $b_1$ and $a_1$, the remaining $(G',H')$ is (isomorphic to) $G_{k-1},H_{k-1}$ and $C':=C\setminus\{a_1b_1, a_1a_2\}$ is a multicut in it. By the induction hypothesis, the minimum multicut of  $(G',H')$ is of size $k-2$.  
	
	Now  either both $a_2 b_2\in C'$ and $a_2 a_3\in C'$  in which case $C'$ is not an inclusonwise minimal multicut for $(G',H')$, since deleting one of $a_1b_1, a_1 a_2$ we already disconnect  the same terminal pairs. So in this case $|C'|>k-2$, and we are done; 
	or one of $a_2b_2, a_2a_3$ is not in $C'$, but then one of $a_1b_1, a_1a_2$ must be in it, since otherwise $b_1$ is not disconnected by $C$ from $b_2$ or from  $a_3$. So in this case $|C|\ge |C'|+1\ge k-1,$ finishing the proof of the first assertion.
	
	For the maximum multiflow value $\mu$, or maximum half-integer and integer multiflow values $\mu_{\rm half}$, $\mu_{\rm int}$, note that the supply edges form a tree, so the set $\mcP$ of paths between the endpoints of demand  edges contains exactly one path $P_e$ for each demand edge $e$, so $|\mcP|=2k-3$. Defining $f(P_{b_{k-1} b_k})=2/3$ and $f(P)=1/3$ for each other path $P\in\mcP$, we have $|f|=2(k-1)/3$.  To prove that this is a maximum flow, note that $e_1:=a_1b_1$ is contained in $P\in\mcP$ if and only if $e_2:=a_2b_2$ is contained in it, so for any multiflow $f$,  
	\[\alpha:=\alpha_f:=\sum_{\{P\in\mcP:e_1\in P\}} f(P)=\sum_{\{P\in\mcP:e_2\in P\}} f(P)\]
	
	
\noindent{\bf Claim}: $\mu\le 2(k-1)/3$, and if $\alpha\le 2/3$, then  $\mu\le 2(k-2)/3 +  \alpha.$

\smallskip
Indeed,  on the one hand,  $\alpha=2/3+\varepsilon$ $(0\le \varepsilon \le 1/3)$ causes  $\alpha'\le 2/3-\varepsilon$  in $(G', H')$ isomorphic to  $(G_{k-1},H_{k-1})$ with max flow value $\mu'$,  after the deletion of $\{a_1,b_1\}$, and then by induction,  $\mu\le  2/3+\varepsilon + \mu'\le 2/3+\varepsilon + 2(k-2)/3 -  \varepsilon=2(k-1)/3$. On the other hand, if  $\alpha=2/3-\varepsilon$ we have by induction
$\mu\le 2/3-\varepsilon + \mu'\le  2/3-\varepsilon + 2(k-2)/3=
2(k-1)/3-\varepsilon$, and the claim is proved, finishing the proof of the second statement.  


The proof of the third assertion, concerning $\mu_{\rm half}$ works similarly. Defining $f(P_{b_i b_{i+1}} )=1/2$ $(i=1,\ldots, k-1)$, and $f(P_{b_1 a_3} )=1/2$, otherwise $f(P)=0$, where $f$ is a feasible half-integer flow,  $|f|=k/2$. To prove  $\mu_{\rm half}\le k/2$ an inductive proof using $\alpha$ works as before:  for $\alpha\ge 1/2$ the stronger statement $\mu\le (k-2)/2 + \alpha$
follows by induction.

Finally, the fourth assertion is immediate from the third one:  $\mu_{\rm int}\le \mu_{\rm half}\le k/2$ is integer, and an integer flow $f$,  $|f|=\lfloor k/2\rfloor$ is also easy to construct. 	
	\qed
	\end{proof}

This theorem provides the best known bounds for several of the defined gaps, and in some cases the best possible bounds that are reviewed below.  

%% file: conclusions_iv.tex
\section{Conclusions}\label{sec:conclusions}

This paper established bounds on the integrality gap of multiflows,  developed approximation algorithms for them, and bounded the flow-cut gap. Applying Theorem~\ref{maxflow mincut} first, then Theorem~\ref{thm:half} to a maximum multiflow, and finally Theorem~\ref{thm:int} to the result, we arrive at the following summary:

\begin{theorem} \label{thm:main}
There exists a $1/4$-approximation algorithm for integer plane multiflow maximisation, with an integer flow-cut gap of $8$; there exists a $1/2$-approximation algorithm for half-integer plane multiflows with a half-integer flow-cut gap of $4$; the flow-cut gap is at most $2$; the  approximation algorithms provide lower bounds to the integrality or half-integrality gap equal to the approximation ratio. 
\end{theorem}  
{\bf a. Multicuts and flow-cut gap}: The minimum multicut is NP-hard to find, but has a PTAS by applying \cite{klein2015correlation} after the transformation of  Lemma~\ref{lem:multicutconnector}.
The (``fractional'') flow-cut gap is in the interval  $[3/2,2]$ by Theorem~\ref{theorem:gap} illustrated in Figure~\ref{fig:gap} and Theorem~\ref{thm:main}; 
the half-integer flow-cut gap is at least $2$ and at most $4$ for the same reason.  For the integer flow-cut gap, a lower bound of  $2$ is  shown by $(\overline{C_4},\overline{2K_2})$ (Footnote~\ref{foot:K4}). The true value of the integer flow-cut gap is thus wide open in the interval $[2,8]$.  

\medskip\noindent
{\bf b. Half-integrality gap}: We do not know the complexity of finding  a half-integer multiflow of maximum value. The worst case ratio between the maximum value of a  half-integer  feasible flow and of a  fractional flow is in the interval $[1/2,3/4]$, again by  Theorem \ref{thm:half} and Theorem \ref{theorem:gap}. It remains an interesting open problem to pin down the exact half-integrality gap in this interval, and the complexity of finding the maximum value of a half-integer flow. 


\medskip\noindent
{\bf c. Integrality gap}: This lies in the interval $[1/4,1/2]$, with lower bound given by Theorem \ref{thm:main}, and upper bound by $(\overline{C_4},\overline{2K_2})$ (Footnote~\ref{foot:K4}). Finally, the worst integer flow/half-integer flow ratio is closed: it is exactly $1/2$, as the same example and Theorem~\ref*{thm:int} show.  


\vskip-0.5cm
\begin{SCtable}
	\begin{tabular}{l|c|c|c|}
		\cline{2-4}
		{\color{blue}SUMMARY}& {\bf fractional} & {\bf half-integer} & {\bf integer} \\
		
		\hline\hskip-1.5pt\vline
		\hbox{  \bf mflow/LP} & 1 &  [1/2,\, 3/4 )\footnotemark[9]   &  [1/4,\,1/2)\footnotemark[10]\\
		\hline\hskip-1.5pt\vline
		\hbox{  \bf mcut/mflow} & (3/2,\,2]\footnotemark[11] & (2,\,4]\footnotemark[12]  & (2,\,8]\footnotemark[11] \\
		\hline\hskip-1.5pt\vline
		\hbox{  \bf complexity} & P  &  {\color{red}  \bf ?}     & NP-hard\\
		\hline
	\end{tabular}
	\caption*{Gaps, complexity, and approxima\-tion ratios for plane multiflows and multicuts }\label{tab:test}
\end{SCtable} 
\vskip-0.7cm\noindent
The approximation ratios for multiflows are the lower bounds of the first row. The second row contains the flow-cut gaps; an upper bound for  the approximation ratio for minimum multicuts is provided by the first upper bound in this row (equal to $2$), but there is also a PTAS \cite{klein2015correlation} for the minimum multicut problem together with an NP-hardness proof. The complexity results (last row) are known (see Introduction) and the Preliminaries, except for half-integer flows. The Gaps and exact approximation ratios are open in the indicated intervals and the maximum half integer plane multiflow problem is also  open. 

\footnotetext[9]{Lower bound: Theorem~\ref{thm:half}\,; upper bound: Theorem~\ref{theorem:gap}\,, see Figure~\ref{fig:gap}\,.}
\footnotetext[10]{Lower bound: Theorem~\ref{thm:half}\,, Theorem~\ref{thm:int}\,; upper bound: $(\overline{C_4},\overline{2K_2})$.}
\footnotetext[11]{Lower bound: Theorem~\ref{theorem:gap}\,, see Figure~\ref{fig:gap}; upper bound: Theorem~\ref{maxflow mincut}}
\footnotetext[12]{Lower bound: $(\overline{C_4},\overline{2K_2})$; upper bound:  Theorem~\ref{thm:main}\,.}

%% file: app_iv.tex
\appendix

\bigskip
\noindent
{\large \bf APPENDIX:  Uncrossing, Laminarity and Integrality}

\smallskip
Uncrossing, laminar families and the related totally unimodular network matrices and integer solutions to linear programs with constraint-matrices having some of these properties are  widely used in combinatorial optimization \cite{schrijver2003combinatorial}, \cite{kortevygen2018combinatorial}, \cite{frank2011connections}. We wish to sketch here the main ideas or references about facts or arguments related to  proofs, algorithms or complexity results of our work. For some of these we provide an adaptation to the special cases and particular purposes of this paper, with the notation and terminology used here.  

The {\bf uncrossing} we use is essentially the same as the one introduced in 
\cite[Section 6.5, p.~235]{LovaszPlummer1986MatchingTheory} for odd cuts (more precisely ``$T$-cuts''). To restate it
recall $G=(V,E)$, $H=(V,F)$, let $G+H$ be planar, and consider cuts in $(G+H)^*=(V^*,E^*\cup F^*)$ of the form $\delta(X):=\delta_{E^*\cup F^*} (X)$. It is easy to check:

\smallskip
{\em  Let $\delta(X), \delta(Y)\subseteq E^*\cup F^*$ $(X,Y\subseteq V^*)$ be two cuts in $(G+H)^*$, both containing exactly one edge of $F^*$, which {\em cross} Then every edge of $E^*\cup F^*$ is covered at most as many times by $\delta (X\cap Y)$ and $\delta(X\cup Y)$, and -- possibly after complementing $X$ say -- both contain exactly one edge of $F^*$.}  Equivalently, {\em either $\delta (X\cap Y)$ and $\delta(X\cup Y)$, or  $\delta (X\setminus Y)$ and $\delta(X\setminus Y)$ contain  exactly one edge of $F^*$, so they ``can replace $\delta(X), \delta(Y)$''.}

Another property not to forget about is that {\em the number (multiplicity) of crossing pairs decreases in this way}.

 It is not difficult to see that the series of replacements, called {\em uncrossing}, is finite even if the cuts are endowed with multiplicity coefficients: then we can uncross $X$ and $Y$ with a  multiplicity equal to the minimum of their coefficients.  It is also true that uncrossing can be executed in polynomial time in a general setting. (See \cite[60.17, 60.3e]{schrijver2003combinatorial} for a proof of finiteness and further references.) 

{\bf Laminarity} is used then throughout the paper. In Section~\ref{sec:multicuts} implicitly, through the WGMV algorithm. In Section~\ref{sec:half} we can use Edmonds and Giles's \cite{edmondsgiles1977} to prove integrality of  \eqref{eq:relax}  and \eqref{eq:c+1}. The closest to our use is Korte and Vygen's ``one-way cut-incidence matrix'' \cite[Theorem 28]{kortevygen2018combinatorial} which directly deduces the total unimodularity of our matrix $M$: replace every $E^*$ by arcs in both directions and put  a capacity of $c(e)$ on both (and infinite capacity on $F^*$) to have  \eqref{eq:relax}, and put $d(u,v), d(v,u)$ on the two arcs to get \eqref{eq:c+1}. These can also be deduced from the rooted tree representations of laminar systems and their total unimodularity  \cite[1.4.1, Section 4.2.2]{frank2011connections} and \cite[Theorem 13.21, pages 213-216]{schrijver2003combinatorial}, since the inequalities in \eqref{eq:relax} and \eqref{eq:c+1} are subpaths of paths from the root, so they define a non-negative network matrix. 

Once network matrices have been identified, total unimodularity is known from Tutte \cite{tutte1965networkunimodular}, included directly in Korte and Vygen's discussion. 
 Hoffman and  Kruskal \cite{hoffmankruskal1956integral} observed the {\bf integer solvability}  of linear programs with totally unimodular coefficient matrices, in polynomial time, and in particular, that of \eqref{eq:relax} and \eqref{eq:c+1}.   
 
After realizing polynomial solvability in multiple ways from known results, one can  develop simple direct algorithms for our special case. 

In fact a maximum integer solution to $Mx\le b, x\ge 0$ and a minimum solution to the dual can be realised by a greedy algorithm for arbitrary $b\in\mathbb{N}_+^{2|E|}$ assigning right hand sides independently of one another to each ordered pair $(u,v)$ and $(v,u)$ such that $uv\in E^*$, in the following way:

We restrict ourselves to   all $1$ costs, as in \eqref{eq:relax} and \eqref{eq:c+1}. 
Take an inclusionwise minimal set $L\in\mcL$ and let $x_L:=b(u_0,v_0)$ be the minimum of $b(u,v)$ among $u, v\in V^*$,  $u\in L$, $v\notin L$.  Deleting row $(u_0,v_0)$ from $M$ and all variables corresponding to a $1$ in this row (involving the deletion of all corresponding columns of $M$ and coordinates of $b$ and $c$), and updating $b$ with the assignments
$b(u,v):=b(u,v) - b(u_0,v_0)$  for all $u \in L, v \notin L$,    continue as follows: 

By induction, the greedy algorithm determines a primal and dual solution of the same value, a relation kept after  the assignment $x_L= b(u_0,v_0)$ and 

{\em if so far $y(u,v)=0$ for each pair $(u,v)$, $u\in L$, $v\in V^*\setminus L$, assign $y(u_0,v_0)=1$.} 

It is straightforward to check that the equality between the primal and dual solution is kept in this way, we leave this to the reader.  The particular structure of our constraints has to be exploited through the following property :   $\mcL(u,v)\supset\mcL (u_0,v_0)$  for all constraints that contain $L$ and do not disappear after the reduction. 